\documentclass[reqno,11pt]{amsart}
\usepackage[utf8]{inputenc}
\usepackage[russian,english]{babel}
\usepackage{amsthm, amsmath, amssymb, yfonts}
\usepackage{epsfig}
\usepackage{amscd}
\usepackage{comment}
\usepackage{mathrsfs}
\usepackage{graphicx}
\usepackage[hyper]{amsbib}
\theoremstyle{theorem}
\newtheorem{thm}{Theorem}[section]
\newtheorem{lem}[thm]{Lemma}

\newtheorem{cor}[thm]{Corollary}
\theoremstyle{definition}

\theoremstyle{remark}
\newtheorem{rem}[thm]{Remark}

\title{Differential-algebraic solutions of the heat equation}
\author{Victor M. Buchstaber, Elena Yu. Netay}

\address{Steklov Institute of Mathematics, Russian Academy of Sciences, ul. Gubkina 8, Moscow, 119991 Russia.}
\email{bunkova@mi.ras.ru (E.Yu.Netay), buchstab@mi.ras.ru (V.M.Buchstaber)}


\begin{document}
\maketitle

\begin{abstract}
In this work we introduce the notion of differential-algebraic ansatz for the heat equation and
explicitly construct heat equation and Burgers equation solutions given a solution of a homogeneous non-linear ordinary differential equation of a special form. The ansatz for such solutions is called the $n$-ansatz, where $n+1$ is the order of the differential equation.
\end{abstract}

\thanks{The work was partially supported by RFBR grants nos. 13-01-12469, 14-01-00012 and Fundamental Research Program of RAS Presidium 19 ``Fundamental problems of nonlinear dynamics in mathematical and physical sciences.''}

\section*{Introduction}

In this work we develop and explicitly describe some results of \cite{FA}. Our main objective is to give a construction of heat equation solutions starting form a solution of a homogeneous non-linear ordinary differential equation or a special ansatz for the solution. The general ansatz we use is called the differential-algebraic ansatz for heat equation and Burgers equation solutions.
Our method gives a wide class of heat equation solutions of a special form that we call $n$-ansatz solutions where $n+1$ refers to the order of the ordinary differential equation. We also describe the corresponding Burgers equation solutions.

The main examples of the differential-algebraic ansatz we have in mind are classical solution of the heat equation in terms of Gaussian normal distribution, and the famous connection between the elliptic theta function $\theta(z, \tau)$ satisfying the heat equation and Weierstrass sigma function. This connection has been described in details in \cite{BB}. Note that in this case our method gives the well-known Chazy-3 equation.

Consider the problem of finding solutions of the $n$-dimensional heat equation
\begin{equation} \label{ndim}
	{\partial \over \partial t} \psi({\bf z}, t) = {1 \over 2} \Delta \psi({\bf z}, t)
\end{equation}
in $\mathbb{R}^n$ with the ansatz $\psi({\bf z}, t) = f\left(|{\bf z}|^2, t\right)$, where $f(x, t)$ is regular at $x = 0$. Here ${\bf z} = \left(z_1, \ldots, z_n \right) \in \mathbb{R}^n$. This gives the equation
\[
	{\partial \over \partial t} f(x,t) = 2 x {\partial^2 \over \partial x^2} f(x,t) + n {\partial \over \partial x} f(x,t)
\]
that results in a recursion on the coefficients $f_k(t)$ for $f(x,t) = \sum_k f_k(t) {x^k \over k!}$:
\[
	f_k'(t) = (2 k + n) f_{k+1}(t).
\]

In the article this problem is solved in the case $n=1$, namely, \eqref{ndim} becomes the classical one-dimensional heat equation and the ansatz restriction assumes $\psi$ is an even function in $z$. 

The work is organized as follows: In section \ref{s1} we describe the general properties of the heat equation and formulate our main results on the differential-algebraic ansatz and the $n$-ansatz for heat equation solutions.
 In \ref{s3} we describe the general properties of the Burgers equation and its relation to the heat equation via the well-known Cole-Hopf transformation.
 In section \ref{ex} we give explicit examples of differential-algebraic and $n$-ansatz solutions of the heat equation and the corresponding Burgers equation solutions.

\section{Heat equation} \label{s1}

\subsection{General remarks.} \text{}

In this paper we address the problem of finding some special solutions of the one-dimensional heat equation, that is 
the equation 
\begin{equation} \label{he}
	{\partial \over \partial t} \psi(z, t) = {1 \over 2} {\partial^2 \over \partial z^2} \psi(z, t).
\end{equation}
Note that, in this equation and in all equations we will consider further on, the variables $t$ and $z$ are graded with $\deg t = 2$, $\deg z = 1$.

The rescaling $t \mapsto 2 \mu t$, $z \mapsto 2 \mu z$ where $\mu \ne 0$ is a constant, $\deg \mu = 0$, brings \eqref{he} to its more general form
\begin{equation} \label{hemu}
	{\partial \over \partial t} \psi(z, t) = \mu {\partial^2 \over \partial z^2} \psi(z, t).
\end{equation}
\begin{comment} In classical mathematics this equation
is associated with Jean Baptiste Joseph Fourier (1768--1830).
\end{comment}
The name of this equation is due to the fact that
it describes the variation in temperature $\psi(z, t)$ over time $t$
in some interval  parametrized by $z$.
In quantum mechanics this equation has the form
of Schrödinger equation (named after Erwin Schrödinger, 1887--1961)
in complex-valued wave function $\psi$. Namely, the equation \eqref{hemu} with $t \mapsto i t$ becomes
\[
	{\partial \over \partial t} \psi(z, t) = \mu i {\partial^2 \over \partial z^2} \psi(z, t).
\]

We will also consider the equation 
\begin{equation} \label{heg}
	{\partial \over \partial t} \psi(z, t) = {1 \over 2} {\partial^2 \over \partial z^2} \psi(z, t) - f(t) \psi(z, t)
\end{equation}
in connection with the problem of finding solutions of the Burgers equation, see section \ref{s3}. Here $f(t)$ is some given function of $t$ representing the radioactive loss of heat.

The heat equation \eqref{he} is linear with respect to $\psi$
and invariant with respect to shifts of arguments.
While our aim is to find solutions of the heat equation, without loose of generality
we may consider only solutions $\psi(z,t)$ of the heat equation
regular at $(z,t) = (0,0)$ that are even or odd functions of $z$.
Any other solution will be a sum of such solutions
up to shifts of arguments.
Moreover, any solution of the heat equation $\psi(z,t)$
may be presented as a sum of an even and an odd function,
each a solution of the heat equation. The same is true for equation \eqref{heg}.

\subsection{Differential-algebraic solutions of the heat equation}\text{}

In this section we introduce an ansatz for differential-algebraic solutions of the heat equation. This ansatz leads to a series of polynomials defined recurrently, such that finding a solution of this recurrence leads to finding solutions of the heat equation. 

In the general case it turns out that explicitly solving this recursion is difficult, so we will impose additional restrictions on the ansatz. 
An example of such a restriction is the $n$-ansatz, for which we impose restrictions on arguments of the functions considered. An alternative restriction is to specify one of the functions involved, namely $h(t)$, in such a way that the recursion simplifies. We will consider both of this approaches.

The most general ansatz we consider in this article is the following ansatz to the one-dimensional heat equation:
\begin{equation} \label{genan}
	\psi(z,t) = e^{-{1\over 2} h(t) z^2 + r(t)} \left(z^\delta + \sum_{k \geqslant 2} \Phi_k(t) {z^{2k+\delta} \over (2k+\delta)!}\right).
\end{equation}
Here the functions $h(t)$, $r(t)$ and $\Phi_k(t)$ have the grading $\deg h(t) = -2$, $\deg r(t) = 0$, $\deg \Phi_k(t) = - 2 k$. This agrees with the grading of the heat equation \eqref{he} considered above.

Heat equation solutions in this ansatz will be called \emph{differential-algebraic solutions}. This solutions have Cauchy boundary conditions $\psi(0,t) = e^{r(t)}, \psi'_z(0,t) = 0$ for $\delta = 0$ and $\psi(0,t) = 0, \psi'_z(0,t) = e^{r(t)}$ for $\delta = 1$. The name is due to the recursion considered below.

The ansatz \eqref{genan} is general in the sense that any regular at $z=0$ even or odd function $\psi(z,t)$ can be presented as a series decomposition
\begin{equation} \label{consti}
	\psi(z,t) = \sum_{k \geqslant 0} \psi_k(t) {z^{2k+\delta} \over (2k+\delta)!}.
\end{equation}

\begin{lem} 
	The condition that \eqref{consti} is a solution of the heat equation \eqref{he} gives the recursion on $\psi_k(t)$:
\[
	\psi_k(t) = 2 \psi_{k-1}'(t).
\]
\end{lem}
The proof is obtained by straightforward substitution of \eqref{consti} into the heat equation.

\begin{rem}
	Note that up to constants defined by \eqref{consti} this recursion is the recursion on the coefficients of $f(x,t)$, see \eqref{ndim}, in the case $n=1$.
\end{rem}

We put $e^{r(t)} = \psi_0(t)$ and $h(t) = -  (1 + 2 \delta)^{-1} e^{- r(t)} \psi_1(t)$ and express $\Phi_k(t)$ in terms of $\psi_k(t)$, $h(t)$ and $e^{r(t)}$ to get \eqref{genan}. This expressions are obtained form $ \Phi_k(t) = e^{- r(t)} \psi_k(t) - P_k(h(t), \Phi_2(t),\ldots, \Phi_{k-1}(t))$ where $P_k$ are some homogeneous polynomials of degree $- 2 k$ with respect to the grading  $\deg h(t) = -2$, $\deg \Phi_k(t) = - 2 k$.
Namely, the first terms in the decomposition above are
\begin{multline*}
	\psi(z,t) = e^{r(t)} z^\delta - (1 + 2 \delta) h(t) e^{r(t)} {z^{2 + \delta} \over (2 + \delta)!} +  \left(\Phi_2(t)  + {(4 + \delta)! \over 8} h(t)^2\right) e^{r(t)} {z^{4 + \delta} \over (4 + \delta)!} +\\ + \left(\Phi_3(t) - (15 + 6 \delta) h(t) \Phi_2(t)  - {(6 + \delta)! \over 48} h(t)^3\right) e^{r(t)} {z^{6 + \delta} \over (6 + \delta)!}  +O(z^8).
\end{multline*}

\begin{lem} \label{lemPhi}
	The condition that \eqref{genan} is a solution of the heat equation \eqref{he} gives $r'(t) = - \left(\delta + {1 \over 2}\right) h(t)$ and the recursion on $\Phi_k(t)$:
\begin{equation} \label{rec}
	\Phi_k(t) =  2 \Phi_{k-1}'(t) + 4  (k - 1) h(t) \Phi_{k-1}(t)  - (2 k + \delta -2) (2 k+ \delta -3 )  \left(h'(t) + h(t)^2\right) \Phi_{k-2}(t).
\end{equation}
The initial conditions are $\Phi_0(t) = 1$,  $\Phi_1(t) = 0$.
\end{lem}
The proof is obtained by straightforward substitution of \eqref{genan} into the heat equation.

\begin{lem}
The condition that \eqref{genan} is a solution of the equation \eqref{heg} gives $r'(t) = - \left(\delta + {1 \over 2}\right) h(t) + f(t)$ and the recursion \eqref{rec} on $\Phi_k(t)$.
\end{lem}

\begin{proof}
	Given a solution of \eqref{he} $\psi(z,t)$ take $e^{g(t)} \psi(z,t)$ where $g'(t) = f(t)$ to get a solution of \eqref{heg}.
\end{proof}

\subsection{The recursion on the coefficients of heat equation solutions.}\text{}

In this section we investigate the recursion \eqref{rec}. We introduce some convenient notations. The motivations for this notations may be found in section \ref{nahe}.

Denote $y_1 = h(t)$, $y_2 = h'(t)$, \ldots, $y_k = h^{(k-1)}(t)$ and $\Phi_k(t) = Y_k(y_1, \ldots, y_k)$.
It follows from lemma \ref{lemPhi} that $Y_k(y_1, \ldots, y_k)$ are homogeneous polynomial of degree $- 2 k$ with respect to the grading $\deg y_k = - 2 k$.
Recursion \eqref{rec} takes the form
\begin{equation} \label{yrec}
	Y_k =  2 \sum_s y_{s+1} {\partial Y_{k-1} \over \partial y_s} + 4  (k - 1) y_1 Y_{k-1}   - (2 k + \delta -2) (2 k+ \delta -3 )  \left(y_2 + y_1^2\right) Y_{k-2}.
\end{equation}

A solution of this recursion gives a solution of the heat equation of the form \eqref{genan} for any $h(t)$.

For instance, we have
\begin{align*}
	Y_0 &=  1, \quad  Y_1 = 0 \\
	Y_2 &=   - (2 + \delta) (1+ \delta)  \left(y_2 + y_1^2\right), \\
	Y_3 &=    - 2 (2 + \delta) (1+ \delta)  \left(y_3 + 6 y_1 y_2 + 4 y_1^3 \right),\\
	Y_4 &=  - 4 (2 + \delta) (1+ \delta)  \left(y_4 + 12 y_1 y_3 + 6 y_2^2 + 48  y_1^2 y_2 + 24 y_1^4 \right)  + (6 + \delta) (5+ \delta) (2 + \delta) (1+ \delta) \left(y_2 + y_1^2\right)^2.
\end{align*}

For a function $h(t)$ let us introduce the operators $\mathcal{L}_k$ and the differential polynomials $\mathcal{D}_k$ in $h(t)$ by
\[
		\mathcal{L}_k = {d \over d t} + 2 k h(t),
\qquad
	\mathcal{D}_1 = \mathcal{L}_{1 \over 2} h,
\qquad
	\mathcal{D}_k = \mathcal{L}_k \mathcal{D}_{k-1}, \quad k >1.
\]
Let $\mathscr{L}$ be a linear operator satisfying the Leibniz rule such that 
\[
	\mathscr{L}h = 1, \qquad \mathscr{L}h^{(k)} = - (k+1) k h^{(k-1)}, \quad k \geqslant 1.
\]

In terms of $y_k = h^{(k-1)}(t)$ we have
\[
	\mathcal{L}_k = \sum_s y_{s+1} {\partial \over \partial y_s} + 2 k y_1, \qquad \mathscr{L} = {\partial \over \partial y_1} - \sum_s (s+1) s y_{s} {\partial \over \partial y_{s+1}},
\qquad
	[\mathscr{L}, \mathcal{L}_k] = 2 k +\mathscr{E},
\]
where $\mathscr{E}$ is the Euler operator $\mathscr{E} =  - 2 \sum_s s y_s {\partial \over \partial y_s}$.  This implies that $\mathscr{L} \mathcal{L}_k P = \mathcal{L}_k \mathscr{L} P$ for any homogeneous polynomial $P$ in $y_1, \ldots, y_k$ with $\deg P = - 2 k$.

The differential polynomials $\mathcal{D}_k$ are homogeneous polynomials in $y_1, \ldots, y_{k+1}$ of degree $-2(k+1)$. We denote $Z_k = \mathcal{D}_{k-1}$, therefore
\[
		Z_2 = \mathcal{L}_{1 \over 2} y_1, \qquad Z_{k} = \mathcal{L}_{k-1} Z_{k-1}, \quad k = 3, 4, \ldots.
\]
Explicitly we have $Z_2  =  (y_2 + y_1^2)$, $Z_3 = (y_3 + 2 y_1 y_2) + 4 y_1 (y_2 + y_1^2)$, $Z_4 = (y_4 + 6 y_1 y_3 + 6 y_2^2 + 12 y_1^2 y_2) + 6 y_1 (y_3 + 6 y_1 y_2 + 4 y_1^3)$ (compare with the first few terms in \eqref{yrec}). For convenience set $Z_0 = Z_1 = 0$.
Note that the set $\{ y_1, Z_2, Z_3, \ldots, Z_k, \ldots \}$ form a multiplicative basis in the space of polynomials in $y_1, \ldots, y_n$.

This observations lead to a proof of the theorem by V.\,M.\,Buchstaber and E.\,Rees that follows.

\begin{lem} \label{lemL}
The operator $\mathscr{L}$ annihilates all $\mathcal{D}_k$:
\[
	\mathscr{L} \mathcal{D}_k = 0 \text{ for any } k.
\]
\end{lem}
\begin{proof}
	We have $\mathscr{L} \mathcal{D}_1 = \mathscr{L} Z_2 = \mathscr{L} (y_2 + y_1^2) = 0$,  $\mathscr{L} \mathcal{D}_k = \mathscr{L} Z_{k+1} = \mathscr{L} \mathcal{L}_{k} Z_{k}$, and form $\deg Z_{k} = - 2 k$ we have $\mathscr{L} \mathcal{L}_{k} Z_{k} = \mathcal{L}_{k} \mathscr{L} Z_{k}$, thus inductively $\mathscr{L} Z_{k+1} = \mathscr{L} Z_{k} = \mathscr{L} Z_2 = 0$.
\end{proof}

\begin{thm}  \label{thmBR}
Let $\mathscr{D}$ be a homogeneous differential polynomial in $h(t)$ of degree $- 2 (n+2)$ with respect to the grading $\deg t = 2$, $\deg h(t) = -2$. 
Then 
\[
	\mathscr{L} \mathscr{D} \equiv 0
\]
if and only if 
\begin{equation} \label{mathD}
	\mathscr{D} = c \mathcal{D}_{n+1} - P_{n}(\mathcal{D}_1, \dots, \mathcal{D}_{n-1})
\end{equation}
for a constant $c$ and a homogeneous polynomial $P_{n}(x_2, \dots, x_{n})$ of degree $n+2$.
\end{thm}

\begin{proof}
	The form \eqref{mathD} is the general form for a homogeneous polynomial in $\mathcal{D}_1, \ldots, \mathcal{D}_{n+1}$ of degree $- 2 (n+2)$. From lemma \ref{lemL} we obtain $\mathscr{L} \mathscr{D} = 0$ for any $\mathscr{D}$ of the form \eqref{mathD}. On the other hand, any homogeneous polynomial $\mathscr{D}$ in $y_1, \ldots y_n$ is a homogeneous polynomial in $y_1, Z_2, \ldots, Z_n$. For its monomial $y_1^{k_1} Z_2^{k_2} \ldots Z_n^{k_n}$ we have $\mathscr{L} y_1^{k_1} Z_2^{k_2} \ldots Z_n^{k_n} = k_1 y_1^{k_1-1} Z_2^{k_2} \ldots Z_n^{k_n}$, therefore $\mathscr{L} \mathscr{D} = 0 \Leftrightarrow k_1 = 0$ for any monomial in $\mathscr{D}$.
\end{proof}

In terms of these operators Lemma \ref{lemPhi} can be reformulated as

\begin{lem} \label{lemlem}
	The condition that \eqref{genan} is a solution of the heat equation gives $r'(t) = - \left(\delta + {1 \over 2}\right) h(t)$ and the recursion on $Y_k$:
\[
	Y_k =  2 \mathcal{L}_{k-1} Y_{k-1}   - (2 k + \delta -2) (2 k+ \delta -3 )  Z_2 Y_{k-2}.
\]
with initial conditions  $Y_0 = 1$, $Y_1 = 0$.
\end{lem}

\begin{cor} \label{corcor}
We have 
\[
	Y_k = -2^{k-2} (2 + \delta) (1 + \delta) Z_k + Q_{k}(Z_2, \dots, Z_{k-2})
\]
Here $Q_{k}(x_2, \dots, x_{k-2})$ are homogeneous polynomials of degree $- 2 k$ defined recurrently by 
\[
	Q_k = 2 \mathcal{L}_{k-1} Q_{k-1} + (2 k + \delta -2) (2 k + \delta -3) Z_2 \left( 2^{k-4} (2 + \delta) (1 + \delta) Z_{k-2} - Q_{k-2} \right)
\]
with initial conditions $Q_2 = Q_3 = 0$. 
\end{cor}

\begin{cor}
$\mathscr{L} Y_k = 0$ for any $k$.
\end{cor}

We see that the coefficients $Y_k = \Phi_k(t)$ are thus homogeneous polynomials of $k-1$ variables $Z_2$, $\dots$, $Z_k$, that are themselves differential polynomials in $h(t)$. This result leads to the introduction of the following ansatz we call the $n$-ansatz.

\subsection{The $n$-ansatz for the heat equation}\text{} \label{nahe}

Let $n \in \mathbb{N}$ be the number of parameters and  $\delta = 0$ or $1$ be the parity.

We say that the differential-algebraic solution of the heat equation is in the \emph{$n$-ansatz} if it has the form
\begin{equation} \label{main}
	\psi(z,t) = e^{-{1\over 2} h(t) z^2 + r(t)} \Phi(z; x_2(t), \dots, x_{n+1}(t)),
\end{equation}
for a vector-function $(x_2(t), \dots, x_{n+1}(t))$ of $t$
where the function $\Phi(z; x_2, \dots, x_{n+1})$
in the vicinity of $z=0$ is given by a series of the form
\begin{equation} \label{form}
	\Phi(z; x_2, \dots, x_{n+1}) = z^\delta + \sum_{k \geqslant 2} \Phi_k {z^{2k+\delta} \over (2k+\delta)!}
\end{equation}
for homogeneous polynomials $\Phi_k = \Phi_k(x_2, \dots, x_{n+1})$ of degree~$- 2 k$ for $\deg x_k = - 2 k$.

In other words, the \emph{$n$-ansatz} is the ansatz of the form \eqref{genan}
with the condition that all $\Phi_k(t)$ can be presented as homogeneous polynomials with constant coefficients in $n$ variables depending on $t$.

A detailed description of this ansatz including a motivation to its introduction are given in~\cite{FA}.
Examples of this ansatz as well as classical solutions of the heat equation rewritten in this ansatz are given in section~\ref{ex}. The results of the next two sections show the interest of introducing this ansatz.

\subsection{Heat dynamical systems and the $n$-ansatz} \text{}

In this subsection we describe the relation of the $n$-ansatz to some special class of polynomial dynamical systems that we will call heat dynamical systems.

Let $p_q(x_2, \dots, x_q)$, where $q = 2, 3, \ldots, n+2$ be a set of homogeneous polynomials with $\deg p_q = - 2 q$ for $\deg x_k = - 2 k$.
We will call \emph{heat dynamical systems} the polynomial dynamical systems in~$x_1(t), x_2(t), \dots, x_{n+1}(t)$ of the form
\begin{align} 
{d \over d t} x_1(t) &=  p_{2}(x_2(t)) - x_1(t)^2, \nonumber \\
{d \over d t} x_k(t) &= p_{k+1}(x_2(t), \dots, x_{k+1}(t)) - 2 k x_1(t) x_k(t), \qquad k = 2, \dots, n, \label{hds}\\
{d \over d t} x_{n+1}(t) &= p_{n+2}(x_2(t), \dots, x_{n+1}(t), 0) - 2 k x_1(t) x_{n+1}(t).   \nonumber
\end{align} 

As an example of such a system one can consider the system satisfied by the parameters of an elliptic curve and the coefficient of its Frobenius-Stickelberger connection, see \cite{BB}. For further examples see section~\ref{ex}.

The name \emph{heat dynamical systems} is due to the following theorem.

\begin{thm}\label{t1}
Let $x_1(t), x_2(t), \dots, x_{n+1}(t)$ be a solution of system \eqref{hds}.

Then one can construct the corresponding even and odd $n$-ansatz solutions of the heat equation \eqref{he} as
\begin{equation} \label{sol}
\psi(z,t) = e^{- \frac12 x_1(t) z^2 + r(t)} \Phi(z; x_2(t), \dots, x_{n+1}(t)),
\end{equation}
where $\delta = 1$ for the odd case and $\delta = 0$ for the even case, ${d \over d t} r(t) = - \bigg(\delta + \frac12\bigg) x_1(t)$,
\[
	\Phi(z;x_2, \dots, x_{n+1}) = z^\delta + \sum_{k \geqslant 2} \Phi_k {z^{2k+\delta} \over (2k+\delta)!},
\]
and $\Phi_k$ are polynomials of $x_2, \dots, x_{n+1}$ with $x_{n+2} = 0$ determined recurrently by
\begin{align*}
\Phi_2 &= - 2 (1 + 2 \delta) p_{2}(x_2), \qquad \Phi_{3} = 2 p_{3}(x_2, x_3) {\partial \over \partial x_2}\,\Phi_{2}, \\
\Phi_{q} &= 2 \sum_{k=2}^{n+1} p_{k+1}(x_2, \dots, x_{k+1}) {\partial \over \partial x_k}\,\Phi_{q-1} + {(2 q + \delta - 3) (2 q + \delta - 2) \over 2 (1 + 2 \delta)}\,\Phi_2 \Phi_{q-2}, \qquad q = 4, 5, 6, \dots
\end{align*}
\end{thm}

\begin{proof}
	This theorem is a corollary of the key theorem in \cite{FA}.
\end{proof}

Let us clarify that the heat dynamical systems do not arise form the heat equation, but are an additional construction used to find $n$-ansatz solutions of the heat equation. Examples can be found in section \ref{ex}.

\begin{rem}
	The polynomial dynamical system \eqref{hds} is homogeneous with respect to the grading considered.
\end{rem}

\begin{rem}
	Due to the grading conditions we have $p_{2}(x_2) = c_2 x_2$, $p_{3}(x_2, x_3) = c_3 x_3$ for some constants $c_2$ and $c_3$.
\end{rem}

\begin{rem}
	The group of polynomial transforms of the form
\begin{equation} \label{polytrans}
x_1 \mapsto x_1, \quad x_2 \mapsto c_2 x_2, \quad x_k \mapsto c_k x_k + q_k(x_2, \dots, x_{k-1}), \qquad k = 3, \dots, n+1,
\end{equation}
where $c_k \ne 0$ are constants and  $q_k(x_2, \dots, x_k)$ are homogeneous polynomials, $\deg q_k = - 2 k$,
acts on the space of heat dynamical systems \eqref{hds}.
\end{rem}

\begin{lem} \label{l1}
Each solution of the form \eqref{sol} of the heat equation obtained in theorem \ref{t1} from a heat dynamical system solution may be obtained by the same construction using a system of the from
\begin{align} 
{d \over d t} x_1 &= x_2 - x_1^2, \nonumber \\
{d \over d t} x_k &= x_{k+1} - 2 k x_1 x_k \qquad \text{for} \qquad  k = 2, \dots, N-1, \label{rhds} \\
{d \over d t} x_{N+1} &= P_{N}(x_2, \dots, x_{N}) - 2 (N+1) x_1 x_{N+1}. \nonumber
\end{align}
Here $N \leqslant n$ and $P_{N}(x_2, \dots, x_{N})$ is a homogeneous polynomial of degree $- 2(N+2)$ for $\deg x_k = - 2 k$. The resulting heat equation solution is in the $N$-ansatz.
\end{lem}

We will call systems of the form \eqref{rhds} \emph{reduced heat dynamical systems}.

\begin{proof}
	This system is obtained from system \eqref{hds} by using the polynomial transforms \eqref{polytrans}. More details on the action of this group can be found in \cite{FA}.
\end{proof}

\begin{rem}
	For $N = n$ system \eqref{rhds} is a special class of heat dynamical systems \eqref{hds} with $p_k(x_2, \dots, x_k) = x_k$, $k = 2, \dots, n+1$, $p_{n+2}(x_2, \dots, x_{n+1}, 0) = P_{n}(x_2, \dots, x_{n})$.
\end{rem}

\subsection{Ordinary differential equations and the $n$-ansatz.}\text{}

We consider the class of equations of the form
\begin{equation} \label{ode}
	\mathcal{D}_{n+1} - P_{n}(\mathcal{D}_1, \dots, \mathcal{D}_{n-1}) = 0.
\end{equation}
Here $P_{n}(x_2, \dots, x_{n})$ is as before a homogeneous polynomial of degree $n+2$.

\begin{rem}
	The left part of equation \eqref{ode} is a homogeneous differential polynomial of degree $- 2 (n + 2)$ with respect to the grading $\deg t = 2$, $\deg h = -2$.
\end{rem}

\begin{thm}\label{t2}
Let $h(t)$ be a solution of the ordinary differential equation \eqref{ode}.

Then one can construct the corresponding even and odd $n$-ansatz solutions of the heat equation \eqref{he} as
\begin{equation*} 
\psi(z,t) = e^{- \frac12 h(t) z^2 + r(t)} \Phi(z; x_2(t), \dots, x_{n+1}(t)),
\end{equation*}
where $\delta = 1$ for the odd case and $\delta = 0$ for the even case, ${d \over d t} r(t) = - \bigg(\delta + \frac12\bigg) h(t)$.
The vector-function $(x_2(t), \dots, x_{n+1}(t))$ is determined by the relation $x_k(t) = \mathcal{D}_{k-1}$ and for the function $\Phi$ we have
\[
	\Phi(z;x_2, \dots, x_{n+1}) = z^\delta + \sum_{k \geqslant 2} \Phi_k {z^{2k+\delta} \over (2k+\delta)!},
\]
where $\Phi_k$ are polynomials of $x_2, \dots, x_{n+1}$ determined recurrently by
\begin{align*}
 \Phi_{1} &= 0, \qquad \Phi_2 = - 2 (1 + 2 \delta) x_2, \\
\Phi_{q} &= 2 \left( \sum_{k=2}^{n} x_{k+1} {\partial \over \partial x_k} + P_{n}(x_2, \dots, x_{n}) {\partial \over \partial x_{n+1}} \right) \Phi_{q-1}+ {(2 q + \delta - 3) (2 q + \delta - 2) \over 2 (1 + 2 \delta)}\,\Phi_2 \Phi_{q-2}, \quad q = 3, 4, 5, \dots
\end{align*}
\end{thm}

\begin{proof}
	This theorem is a corollary of theorem \ref{t1}. We have $x_1(t) = h(t)$ and the expressions for $x_k(t)$ for $k \geqslant 2$ follow from lemma \ref{l1}.
\end{proof}

Therefore given a polynomial $P_{n}$ for some $n$ and a solution of the ordinary differential equation \eqref{ode} we obtain an even and an odd $n$-ansatz solutions of the heat equation.

\section{Burgers equation} \label{s3}

\subsection{General remarks.}  \text{}

\emph{The Burgers equation}
\begin{equation} \label{Bur}
v_t + v v_z = \mu v_{zz}
\end{equation} 
is named after Johannes M.\,Burgers, 1895--1981.
Here $v  = v(z,t)$ and $\mu$ is constant, $\deg \mu = 0$.
This equation occurs in various areas of applied mathematics, for example
fluid and gas dynamics, acoustics, traffic flow.
It is used for description of wave processes with velocity $v$
and viscosity coefficient $\mu$.

The rescaling $t \mapsto \nu t$, $z \mapsto \nu z$ where $\nu \ne 0$ is a constant, $\deg \nu = 0$, brings \eqref{Bur} into the Burgers equation with $\mu \mapsto \nu \mu$.

For $\mu =0$ we obtain the Hopf equation
(named after Eberhard F.\,Hopf, 1902--1983).
It is the simplest equation describing discontinuous flows or flows with shock waves.

\emph{The Cole-Hopf} transformation of a function $\psi(z, t)$ is
\[
v(z, t) = - 2 \mu {\partial \ln \psi(z, t) \over \partial z}. 
\]
\begin{lem}[see \cite{Hopf}]
Let $\psi(z, t)$ be a solution of the equation 
\[
	{\partial \over \partial t} \psi(z, t) = \mu {\partial^2 \over \partial z^2} \psi(z, t) - f(t) \psi(z, t) 
\]
for some $f(t)$. Then $v(z, t)$ is a solution of the Burgers equation.
\end{lem}

The proof of this lemma follows form the identity 
\[
v_t +  v v_z -  \mu v_{zz} = - 2 \mu {\partial \over \partial z} \left({\psi_t - \mu \psi_{zz} \over \psi} \right). 
\]

\begin{cor} \label{cor3}
Let $\psi(z, t)$ be a solution of the heat equation \eqref{he} or \eqref{heg}. Then 
\[
v(z, t) = - {\partial \ln \psi(z, t) \over \partial z}. 
\]
is a solution of the Burgers equation with $\mu = {1 \over 2}$.
\end{cor}

	The Burgers equation is not linear, therefore if we consider only even or only odd solutions with respect to $z$ we will not obtain the general solution. Yet we may observe that if $\psi(z, t)$ is a solution of the heat equation even or an odd with respect to $z$, then its Cole-Hopf transformation $v(z,t)$ is a solution of the Burgers equation which is odd with respect to $z$. Let us focus on the problem of finding odd solutions of the Burgers equation.
	
	Another related note is that an even heat equation solution regular at $(z, t) = (0, 0)$ gives a Burgers equation solution regular at $(z, t) = (0, 0)$ while for an odd heat equation solution the same is not true.

\subsection{$N$-ansatz for the Burgers equation.} \text{}
	
The image of the $n$-ansatz for the heat equation under the Cole-Hopf has the form
\[
	v(z,t) = - {\delta \over z} + h(t) z - \Psi(z; x_2(t), \dots, x_{n+1}(t)),
\]
for a vector-function $(x_2(t), \dots, x_{n+1}(t))$ of $t$
where the function $\Psi(z; x_2, \dots, x_{n+1})$
in the vicinity of $z=0$ is given by a series of the form
\[
	\Psi(z; x_2, \dots, x_{n+1}) = \sum_{k \geqslant 2} {\Psi_k \over (2 \delta k + 1)} {z^{2k-1} \over (2k-1)!}
\]
for homogeneous polynomials $\Psi_k = \Psi_k(x_2, \dots, x_{n+1})$ of degree~$- 2 k$ for $\deg x_k = - 2 k$.

Here $\delta$ is as before $0$ or $1$. 

The relation with the heat equation $n$-ansatz \eqref{main} is
\[
	\Psi_k = \Phi_k + Q_k(\Phi_2, \Phi_3, \ldots, \Phi_{k-2})
\]
for some homogeneous polynomials $Q_k$ of degree~$- 2 k$ for $\deg \Phi_k = - 2 k$.

\begin{rem}
	Note that for $\delta = 0$ if we put $h(t) = - \Psi_1$ we simply obtain the series decomposition for an odd function of $z$ regular at $z = 0$
	\[
		- v(z,t) =  \sum_{k \geqslant 1} \Psi_k {z^{2k-1} \over (2k-1)!}.
	\]
	The $n$-ansatz condition in this case is that $\Psi_k$ for $k \geqslant 2$ are all polynomials of $n$ graded variables $x_k$.

\end{rem}

\subsection{Differential-algebraic ansatz for the Burgers equation}\text{}

The image of the differential-algebraic ansatz for an even solution of the heat equation under the Cole-Hopf transform has the form
\[
	v(z,t) = h(t) z -  \sum_{k \geqslant 2} \Psi_k(t) {z^{2k-1} \over (2k-1)!},
\]
where $h(t)$ and $\Psi_k(t)$ are functions of $t$ with $\deg h(t) = - 2$, $\deg \Psi_k(t) = - 2 k$.

Therefore any odd solution of the Burgers equation regular at $z=0$ is an image of an even solution of the heat equation under the Cole-Hopf transform.

The case  $\delta = 1$ can be considered in a similar way.

\section{Examples of $n$-ansatz solutions of the heat equation}  \label{ex} 

\subsection{Rational solutions of heat dynamical systems} \text{}

For special $P_n$ equation \eqref{ode} admits a series of rational solutions of the form
\[
	h(t) = {1 \over n+1} \sum_{k=1}^{n+1} {\alpha_k \over \alpha_k t - \beta_k} 
\]
for constants $(\alpha_k:\beta_k) \in \mathbb{C}P^1$.

We have
\[
	{d^m \over d t^m} h(t) =   {(- 1)^m m! \over n+1} \sum_{k=1}^{n+1} {\alpha_k^{m+1} \over (\alpha_k t - \beta_k)^{m+1}} 
\]
and ${d^{n+1} \over d t^{n+1}} h(t)$ is expressed in terms of ${d^m \over d t^m} h(t)$ for $m < n+1$ using classical formulas. This expression is homogeneous and gives equation \eqref{ode}.

Some examples of the corresponding heat equation solutions will be given in the following sections.

\subsection{$0$-ansatz solutions} \label{ses} \text{}

Let us consider the case with an empty set of parameters $x_k$.

According to theorem \ref{t1}, the solutions of the heat equation in the $0$-ansatz have the form
\[
	\psi(z,t) = e^{-{1\over 2} h(t) z^2 + r(t)} z^\delta,
\]
where  ${d \over d t} r(t) = - \bigg(\delta + \frac12\bigg) h(t)$ and system \eqref{hds} with $x_1(t) = h(t)$ yelds
\[
{d \over d t} h(t) =   - h(t)^2.
\]
This equation coincides with $\mathcal{D}_1 = 0$.
Its general solution is
\[
	h(t) = {\alpha \over \alpha t - \beta}, \quad \text{where} \quad (\alpha:\beta) \in \mathbb{C}P^1.
\]

\begin{thm}
	The general $0$-ansatz solution of the heat equation is 
\begin{equation} \label{n0}
	\psi(z,t) = \left({\alpha \over \alpha t - \beta}\right)^{{1 \over 2} + \delta} \exp\left(- {\alpha z^2 \over 2 (\alpha t - \beta)} + r_0\right) z^\delta
\end{equation}
where $\delta = 1$ for an odd function and $\delta = 0$ for an even function, $(\alpha:\beta) \in \mathbb{C}P^1$ and $r_0 \in \mathbb{C}$ are constants.
\end{thm}

\begin{rem}
	For $\delta = 0$ and ${\beta \over \alpha} < t < \infty$ the $0$-ansatz solution gives the Gaussian density distribution on the real axis.
\end{rem}

For the general $0$-ansatz solution \eqref{n0} and $\mu = {1 \over 2}$
the Cole-Hopf transformation gives a rational solution of the Burgers equation \eqref{Bur}:
\[
v(z,t) =  { \alpha z \over \alpha t - \beta} - {\delta \over z}.
\]

The following image represents solution \eqref{n0} and its Cole-Hopf transformation in the case $\delta = 0$, $\beta = 0$, $r_0 = 0$ for $t = {1 \over 16}$, $t = {1 \over 4}$ and $t = 1$.
In this case the boundary conditions are $\psi(0,t) =t^{-{1 \over 2}}, \psi'(0,t) = 0$.
\begin{center}
\includegraphics[scale=0.35]{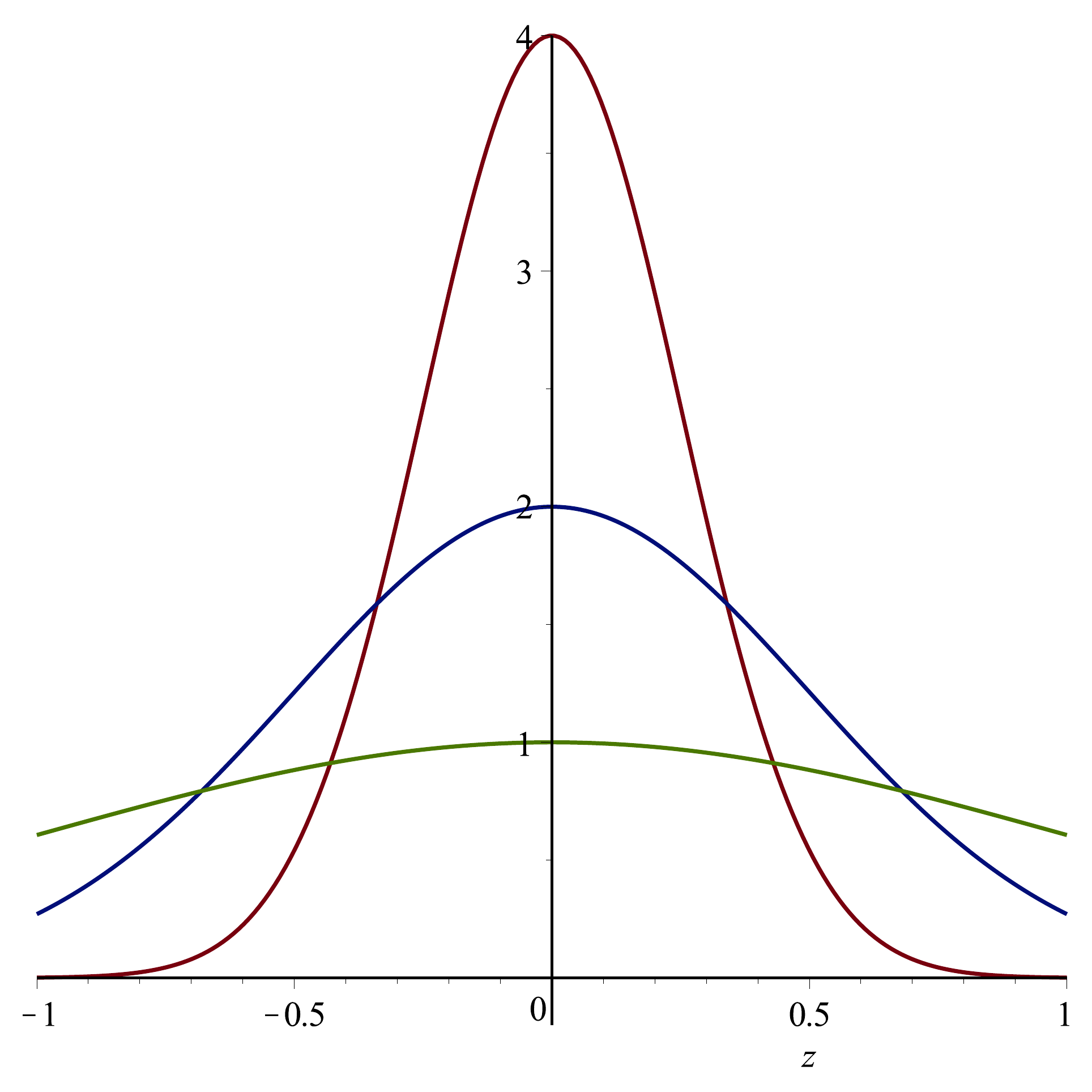}
\includegraphics[scale=0.35]{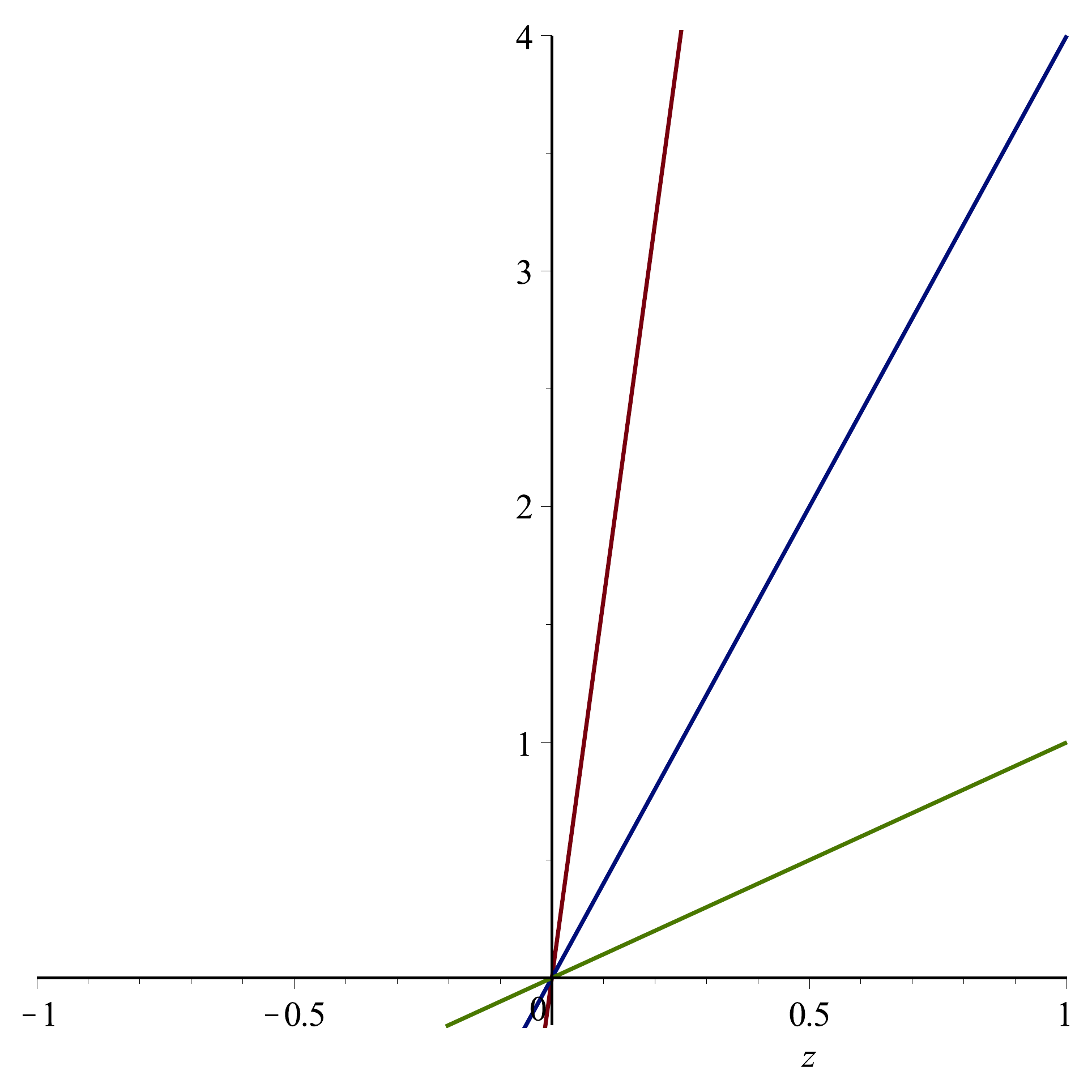}
\end{center}

The following image represents solution \eqref{n0} and its Cole-Hopf transformation in the case $\delta = 1$, $\beta = 0$, $r_0 = 0$ for $t = {1 \over 16}$, $t = {1 \over 4}$ and $t = 1$.
In this case the boundary conditions are $\psi(0,t) =0, \psi'(0,t) = t^{-{3 \over 2}}$.
\begin{center}
\includegraphics[scale=0.35]{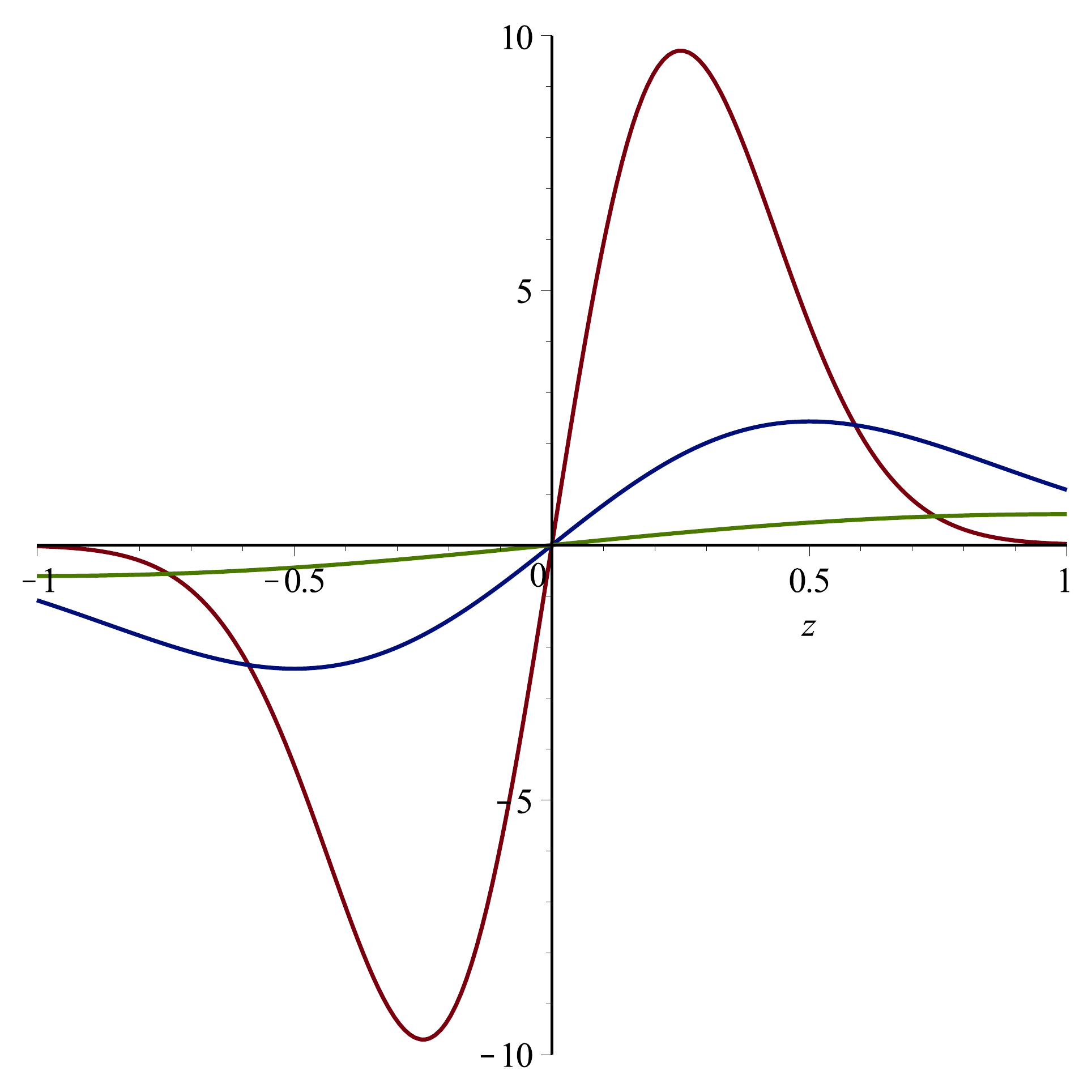}
\includegraphics[scale=0.35]{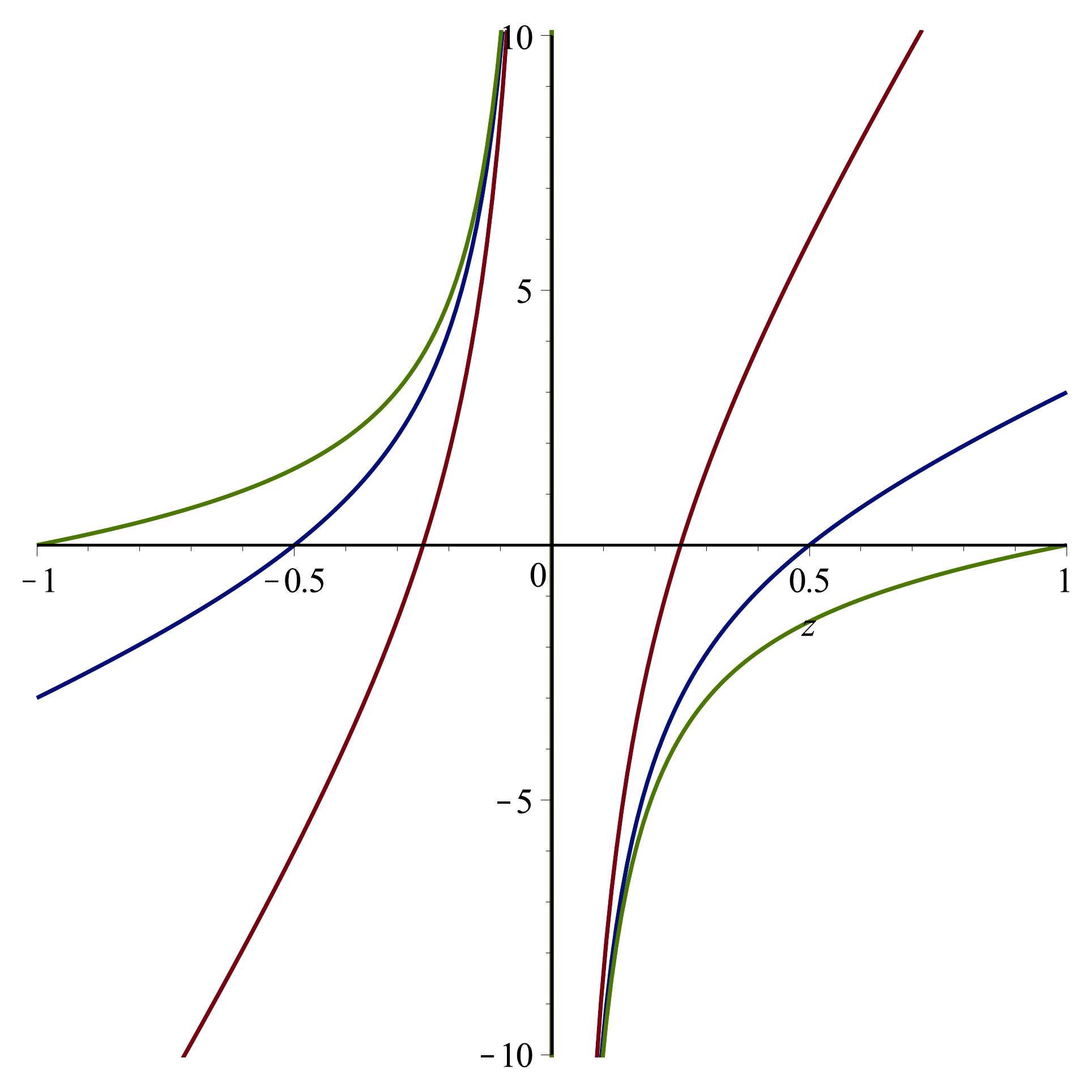}
\end{center}

\subsection{$1$-ansatz solutions} \label{sef} \text{}

According to theorem \ref{t2}, in this case the solutions of the heat equation in the $1$-ansatz have the form
\[
\psi(z,t) = e^{- \frac12 x_1(t) z^2 + r(t)} \Phi(z; x_2(t)),
\]
where  ${d \over d t} r(t) = - \bigg(\delta + \frac12\bigg) x_1(t)$,
we have $P_1 = 0$ due to the grading 
and system \eqref{rhds} yields
\[
{d \over d t} x_1 = x_2 - x_1^2, \qquad
{d \over d t} x_{2} = - 4 x_1 x_2.
\]
This system is equivalent to the equation $\mathcal{D}_2 = 0$ on $x_1(t) = h(t)$, that is
\begin{equation} \label{odd1}
	h'' + 6 h h' + 4 h^3 = 0.
\end{equation}
Differentiate \eqref{odd1} and put $y(t) = 2 h(t)$
to get Chazy-4 equation:
\[
	y''' = - 3 y y'' - 3 y'^2 - 3 y^2 y'.
\]
The general solution of \eqref{odd1} has the form
\begin{equation} \label{h2}
	h(t) = {1 \over 2} \left({\alpha_1 \over \alpha_1 t - \beta_1} + {\alpha_2 \over \alpha_2 t - \beta_2}\right), \quad (\alpha_1: \beta_1), (\alpha_2: \beta_2) \in \mathbb{C}P^1.
\end{equation}

\begin{thm}
	The general $1$-ansatz solution of the heat equation is a product of two entire functions:
\begin{equation} \label{n1}
	\psi(z,t) = \mathcal{G}(z,t) \Phi(z; x_2(t)).
\end{equation}
The explicit expressions for this functions are
\begin{align*}
\Phi(z; x_2(t)) &= {z}^{\delta} \sum_{m=0}^\infty \frac{ \, \Gamma({3\over4} + {\delta \over 2})}{m! \, \Gamma(m+ {3\over4} + {\delta \over 2})} (- 1)^m x_2(t)^m {\left(\frac{ z } {2}\right)}^{\!\!4 m}\!\!\!,
\\
	\mathcal{G}(z,t) &= \prod_{k=1,2 \atop \alpha_k \ne 0} \left({\alpha_k \over \alpha_k t - \beta_k}\right)^{ {1 + 2 \delta \over 4}}  \exp\left(  - {z^2 \over 4} \left({\alpha_1 \over \alpha_1 t - \beta_1} + {\alpha_2 \over \alpha_2 t - \beta_2}\right) + r_0 \right).
\end{align*}
Here
\[
	x_2(t) = - {1 \over 4} \left({\alpha_1 \over \alpha_1 t - \beta_1} - {\alpha_2 \over \alpha_2 t - \beta_2}\right)^2
\]
 $\delta = 1$ for an odd function and $\delta = 0$ for an even function, $(\alpha_k:\beta_k) \in \mathbb{C}P^1$ for $k = 1,2$ and $r_0 \in \mathbb{C}$ are constants.
\end{thm}

\begin{proof}
According to Theorem \ref{t1}, the function $\Phi(z; x_2)$ has the form \eqref{form} where $\Phi_k$ are polynomials of $x_2$ determined recurrently by the relations
\begin{align*}
\Phi_2 &= - 2 (1 + 2 \delta) x_2, \qquad  \Phi_{q} = 0, \quad q = 1,3,5, \dots \\
\Phi_{q} &=  {(2 q + \delta - 3) (2 q + \delta - 2) \over 2 (1 + 2 \delta)}\,\Phi_2 \Phi_{q-2}, \qquad q = 4, 6, 8, \dots
\end{align*}
The solution to this recursion is stated above. Expressing $x_2(t)$ and $r(t)$ in terms of $h(t)$ given by \eqref{h2}, we get the statement of the theorem.
\end{proof}

The following images represent solution \eqref{n1} in the case of different parameters:\\
1) $\delta = 0$, $\alpha_1 = 1$, $\beta_1 = 0$, $\alpha_2 = 0$, $r_0 = 0$ for $t = {1 \over 16}$ (red), $t = {1 \over 4}$ (green), $t = 1$ (blue).\\
In this case $\psi(0,t) =t^{ -{1 \over 4}}$.

2) $\delta = 0$, $\beta_1 = 0$, $\alpha_2 = \beta_2$, $r_0 = 0$, for $t = {1 \over 25}$ (red), $t = {1 \over 2}$ (yellow), $t = {99 \over 100}$ (green), $t = {101 \over 100}$ (blue), $t = 2$ (violet).\\
In this case $\psi(0,t) =t^{- {1 \over 4}}  \left(t - 1\right)^{- {1 \over 4}}$.

\begin{center}
	\includegraphics[scale=0.35]{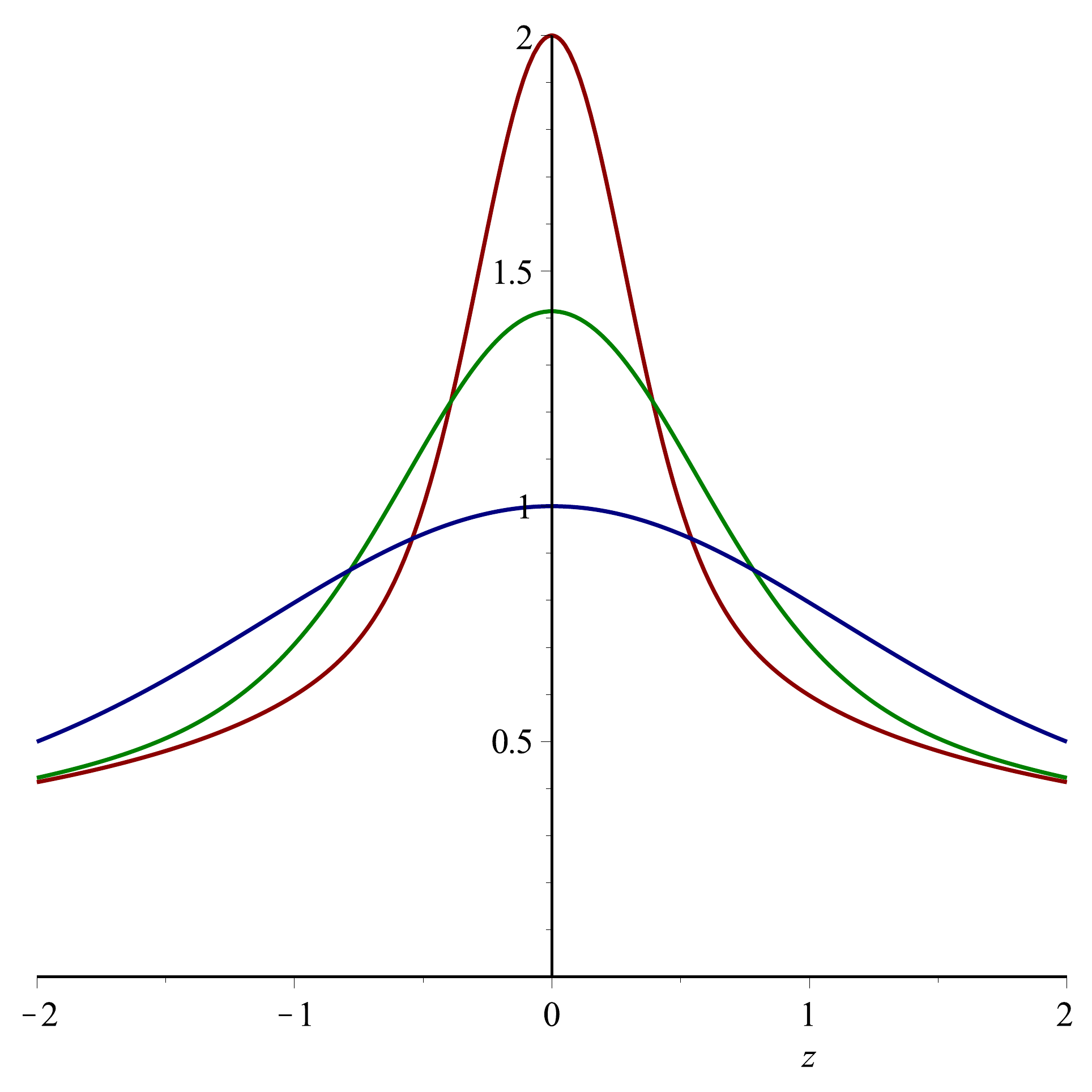} 
	\includegraphics[scale=0.35]{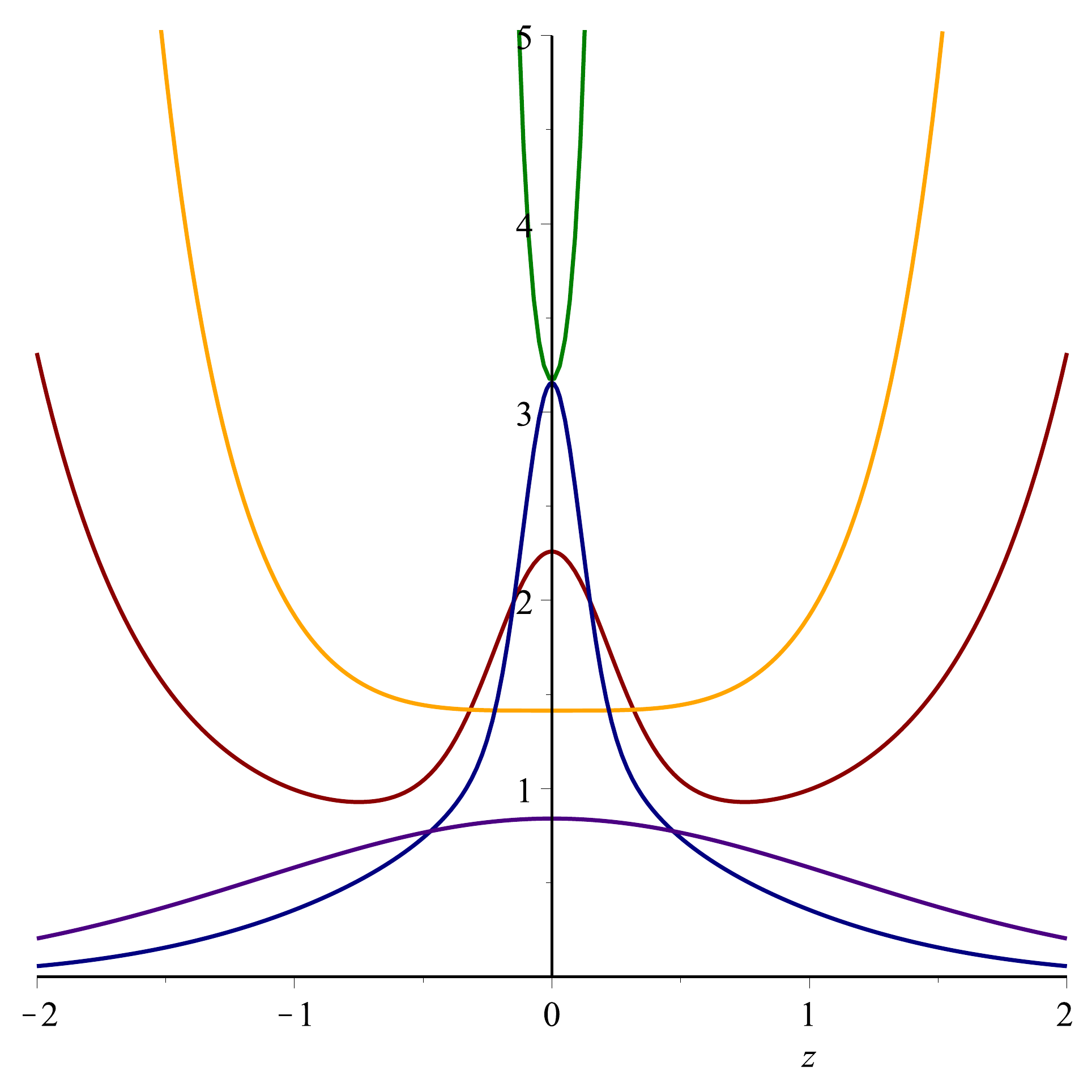}
\end{center}

\subsection{$2$-ansatz solutions} \text{}

For a full description of the case $n=2$ see \cite{FA}. Note that this case appear heat equation solutions in trems of elliptic functions.

\bibliographystyle{amsplain}

\end{document}